\documentclass[11pt]{article}

\usepackage{wrapfig}
\usepackage{mathtools}
\usepackage[dvipsnames]{xcolor}
\usepackage{graphicx} 
\usepackage{amsthm}
\usepackage[margin=1in]{geometry}

\usepackage[colorlinks=true,linkcolor=magenta,citecolor=magenta]{hyperref}
\usepackage{cleveref}

\DeclareMathOperator*{\Probability}{Pr}
\newcommand{\prob}[1]{\Pr\left[#1\right]}
\newcommand{\Prob}[2]{\Probability_{#1}\left[#2\right]}

\newcommand{\psucc}{P_{succ}}

\newcommand{\psuccun}{\psucc}
\newcommand{\pfail}{P_{fail}}

\newcommand{\DRepsopt}[2]{k_{\max}(#2,#1)}

\newcommand{\psuccub}{\hat{P}_{succ}}
\renewcommand{\vec}[1]{\mathbf{#1}}

\newtheorem{theorem}{Theorem}
\newtheorem{lemma}{Lemma}
\newtheorem{assumption}{Assumption}
\newtheorem{remark}{Remark}
\newtheorem{corollary}{Corollary}
\newtheorem{example}{Example}
\newtheorem{definition}{Definition}
\newtheorem{observation}{Observation}

\title{
Reward Schemes and Committee Sizes in Proof of Stake Governance
}

\author{Georgios Birmpas \thanks{University of Liverpool, Email:  \texttt{G.Birmpas@liverpool.ac.uk}} \and Philip Lazos \thanks{IOG, Email:  \texttt{plazos@gmail.com}} \and Evangelos Markakis \thanks{Athens University of Economics and Business, and IOG, Email:  \texttt{markakis@gmail.com}} \and Paolo Penna \thanks{IOG, Email:  \texttt{paolo.penna@iohk.io}}}

\date{}

\begin{document}

\maketitle

\begin{abstract}
In this work, we investigate the impact of reward schemes and committee sizes on governance systems over blockchain communities. We introduce a model of elections with a binary outcome space, where there is a ground truth (i.e., a ``correct" outcome), and where stakeholders can only choose to delegate their voting power to a set of delegation representatives (DReps). Moreover, the effort (cost) invested by each DRep positively influences both (i) her ability to vote correctly and (ii) the total delegation that she attracts, thereby increasing her voting power. 
This model constitutes the natural counterpart of delegated proof-of-stake (PoS) protocols, where  delegated stakes are used to elect the block builders.    

As a way to motivate the representatives to exert effort, a reward scheme can be used based on the delegation attracted by each DRep. We analyze both the game-theoretic aspects and the optimization counterpart of this model. Our primary focus is on selecting a committee that maximizes the probability of reaching the correct outcome, given a fixed monetary budget allocated for rewarding the delegates. Our findings provide insights into the design of effective reward mechanisms and optimal committee structures (i.e., how many DReps are enough)  in these PoS-like governance systems.
\end{abstract}

\newpage

\section{Introduction}\label{sec:intro}

Our work falls under the broader topic of selecting an appropriate set of representatives out of a voting population. This has clearly been a prominent research agenda in social choice theory over the years and has been already investigated from various angles. 
As indicative directions, the performance of randomly selected committees has been extensively studied and there also exist various formulations of finding the optimal number of representatives either as an optimization problem or via game-theoretic models (described also in our related work section). At the same time this is also complemented by empirical research and the study of real world practices, spanning a horizon of several decades, see e.g. \cite{T72}.

We focus on addressing such questions for  governance systems in Proof-of-Stake (PoS) blockchain protocols (see e.g. \cite{larangeira2023security}). Several blockchain communities have already implemented or are currently designing governance policies, where stakeholders can propose a referendum on any relevant issue, which can then lead to an election. We believe there are some important aspects that can jointly differentiate such elections from other more traditional settings. First, in blockchain communities, voting among stakeholders is usually a weighted voting process, with the voting power corresponding to the stake owned by each user. This moves away from the classic ``one person-one vote" paradigm, which cannot be enforced due to the anonymity of users (someone could vote with several identities by splitting her stake). Secondly, in some blockchains, elections are implemented only by delegating voting power to representatives (known a priori) who will then vote with a weight equal to their total delegation they collected. This is done both in order to avoid having a huge number of transactions in the long run (a delegation can remain valid for future elections too, for as long as the user wants) but also to give the option to users who are not well-informed on an election topic to transfer their rights to someone that they trust their opinion. Third, it has been acknowledged that the users who act as representatives should be given some monetary compensation. The reason for this is two-fold. Representatives need to exert an effort  
 to advertize their opinion and attract voters. But more importantly, the elections under consideration may often concern a technical topic (like protocol parameter changes), where the representatives may need to spend time so as to become more informed and shape an opinion. Crucially, even after the vote, the system has no way to detect if the answer is correct.

The features highlighted in the previous discussion, motivate various interesting questions for the design of appropriate policies. In particular, an important question that arises is how to design a reward scheme for the representatives, given an available budget.
What are the relevant parameters that the reward should depend on? Ideally, we would like a reward scheme to induce good quality Nash equilibria, meaning that the representatives are incentivized to exert a sufficient amount of effort so that their vote contributes to making a good decision for the blockchain protocol. Therefore the rewards need to account for the fact that effort can be costly. At the same time, another relevant question is to understand how many representatives can be enough under such a scenario. Qualitatively, what we are interested in is whether a relatively small set of representatives can be sufficient or whether a large number of them is necessary to ensure a good election outcome, i.e., to ensure that the weighted majority of the representatives vote for the ground truth.

\subsection{Contribution}
Our work is motivated by the ongoing design of the governance system in an actual blockchain community, namely of the Cardano cryptocurrency \cite{1694}.
In Section \ref{sec:model} we introduce a game-theoretic model for capturing the main aspects of the elections under consideration. We focus on the scenario where the outcome space is binary and there is a {\it ground truth}, i.e., there is a correct outcome (say for the long-term evolution of the protocol), not a priori known to the representatives. The representatives can exert effort in order to find out the correct outcome which however comes at a cost (for information acquisition). At the same time, the exerted effort leads to a higher level of attracted delegations and in turn to (potentially) higher rewards, as we focus on reward schemes that depend on the volume of delegations (similarly to  delegated PoS protocols).

In Sections \ref{sec:warmup} and \ref{sec:candmech}, we analyze natural reward schemes under this setting regarding their equilibria. In Section \ref{sec:warmup}, we demonstrate, to our surprise, that perhaps the most natural rule where the representatives split the total budget in a proportional manner to their attracted delegations, is not so appropriate. The reason is that it induces low quality equilibria where not enough effort is made. In Section \ref{sec:candmech}, we advocate the use of a better mechanism where rewards are given only to voters who reach a desired threshold of delegations (and thus up to $k$ of them, for some parameter $k$). We characterize the set of pure equilibria, and show that under any equilibrium, the representatives exert a significant effort, and hence contribute towards electing the correct outcome. We also comment on relevant variants of this mechanism.

As the scheme of Section \ref{sec:candmech} imposes an upper bound on the representatives who will make an effort at equilibrium, this motivates the algorithmic question of how to select the number of representatives. We investigate this in Section \ref{sec:optsol}, as a budget constrained problem. We study various classes of cost functions for the exerted effort, including concave, convex and concave-convex costs and highlight the different behavior of the optimal solution under these classes. One of our main conclusions is that in many cases, a small number of representatives suffices for achieving a good quality outcome, e.g., this holds for concave costs and also for concave-convex functions. Another interesting finding is that in the convex domain the answer is heavily dependent on the available budget. Finally, we also demonstrate our findings in Section \ref{sec:S-shape} for a particular class of concave-convex cost functions that is motivated by works on experimental psychology.

\subsection{Related work}\label{sec:rw}

The topic of incentivizing committee members or delegates in elections to exert more effort has recently attracted attention partly due to applications over blockchain protocols and partly due to the overall rise of proxy voting and liquid democracy. 
The works most related to ours along these lines are \cite{GMT22}, \cite{CS23} and \cite{gersbach2024we}. In \cite{GMT22}, a similar cost model to ours is used for acquiring information over an election topic. There are however substantial differences in most other modeling aspects. In their work, a committee is chosen randomly among a given population, and with the same voting power per member, whereas in our case the voting power depends on the effort exerted. Secondly, they consider a different game in which 
the rewards are dependent on the tally difference between the two alternatives, which is quite different from our model, where the reward depends on the attracted delegation. Furthermore, the payments are transfers from the remaining population and not by some external source. In a recent follow up work \cite{gersbach2024we}, a similar model to \cite{GMT22} is studied, but where the monetary transfers depend on observed information acquisition costs. In both of these works, the particular structure of the underlying equilibria is studied and 
despite the differences with our setting, the conclusions made on the appropriate committee size are qualitatively of a similar flavor. Finally, in the work of \cite{CS23}, a different model is considered where the voters are categorized into well-informed and mis-informed agents, with different cost functions each. The reward scheme considered there is also different, with no delegation involved, and where the payments are dependent on the fraction of other voters who voted the same alternative.  

A different approach is taken in \cite{AG12} for determining a reward scheme and a committee size. Namely, a mechanism design model is presented where a committee is picked at random from the population, and where the rewards are obtained as the outcome of a truthful mechanism. For the question of finding the optimal number of representatives, there have also been other attempts that are not based on rewarding the voters. A game-theoretic model along this direction is presented in \cite{Martinelli06}. In \cite{RLH22}, a different model is studied where the optimal committee set is derived as the one maximizing the probability of voting for the correct outcome, given competence levels from some distribution. A similar idea is also used in \cite{MX18} under constraints on the feasible sets of delegates. Yet another approach is explored in \cite{ZP20} by adding the dimension of a social network for defining an optimal set of representatives in elections.   

When there is no a priori ground truth, alternative methodologies within proxy voting have also been considered for selecting a committee size. These approaches are based on optimizing the total welfare of the electorate. As an example, the performance of the {\it Sortition} method (randomly pick a subset of the voters of a given size) is studied in \cite{MST21}. The work  \cite{PS20} considers the performance of proxy voting, focusing on understanding when the proxy-elected outcome coincides with the outcome of direct voting. Finally, in \cite{AFLMP24} welfare guarantees are provided for a small number of representatives under incomplete preferences.

The topic of incentivizing effort has been extensively studied in economics within the 
field of contract theory \cite{Ross73}. The models there typically involve a principal who can offer a contract to an agent for performing some task. Recently there has been a renewed interest in such problems from an algorithmic viewpoint, and we refer to \cite{DFC24} for an upcoming survey. These models however are only distantly related  to our work, as the effort there is not tied to attracting delegations or voting power. In our model, rewards cannot depend on the ``correct answer/success" as the latter cannot be verified. This contrasts with prior contract theory (principal-agent) models, where reward schemes are tied to the success of the project.

\section{Our delegation model}\label{sec:model}

We are considering a voting scenario with a binary choice, consisting of a `good' and a `bad' outcome. This is initially unknown to laypeople voters, but delegates are ready to step up, using their expertise and effort to steer the community towards the right choice.

Suppose that we have $n$ delegation representatives (DReps) competing for votes. Each delegate can choose to exert some effort $x_i$, which leads to two positive effects: an increased chance of voting for the right outcome and an increased number of voters delegating to $i$. 
In particular, an effort $x_i\in [0, 1/2]$ leads to a probability of $p_i = 1/2 + x_i$ for $i$ to vote for the right outcome. 
At the same time, the DRep manages to accumulate delegation (and subsequently voting power) equal to
\begin{equation}\label{def:delegated}
w_i = \frac{x_i}{\sum_j x_j},
\end{equation}
with $\vec{w} = (w_1, \ldots, w_n).$ 
If none of the representatives exert any positive effort, then we assume $w_i=0$ for all $i$. 
The underlying assumption here is that a more informed and knowledgeable DRep is more likely to attract voters as well (a different interpretation is that voters follow DReps based on their long term performance on voting for the `supposedly good' outcome). We assume that, even after the vote, the system has no way to detect if the answer is correct (thus rewards cannot directly depend on this information).

The effort $x_i$ also comes with a cost, described by a cost function\footnote{Following the relevant literature \cite{GMT22,gersbach2024we}, we also consider a common cost function for all voters.} $c(x_i)$, for which we assume that it is strictly increasing, continuous and differentiable. We will consider various cases for the cost function in the sequel, such as linear, convex, concave, and concave-convex types observed in experimental psychology.

We focus on reward mechanisms that provide a monetary payment to each DRep based on the delegation vector $\vec{w}$. Such mechanisms are easy to implement as the total attracted delegation for each DRep can be measured (the effort $x_i$ cannot). If $f_i$ is the reward function for the payment to DRep $i$, her final utility is
\begin{equation}
u_i = f_i\left(\vec{w}\right) - c(x_i) \ . 
\end{equation}
\paragraph{Equilibria.}
A pure Nash equilibrium for a particular combination of reward functions $f_i$, and cost function $c$ is defined as an effort vector $\vec{x} \in [0, 1/2]^n$ such that for any player $i$ and $x_i' \neq x_i$ we have that:
\begin{equation}
    u_i(\vec{x}) \ge u_i(x_i', \vec{x_{-i}})\ .
\end{equation}

\paragraph{Budget Constraint.}
The total rewards given should be limited, such that given an available budget $B$:
\begin{equation}\label{eq:budget_constraint}
    \sum_{i=1}^n f_i(\vec{w}) \le B \ .
\end{equation}

\paragraph{Objective.}
Our social objective is to  maximize the probability that the `good' outcome is selected by the election.
Since each DRep $i$ has a voting power equal to $w_i$, the right outcome is selected only when the total weight of the DReps voting correctly is at least $1/2$. For the case where it is exactly $1/2$, we assume a random tie-breaking, so that the correct outcome is selected with probability $1/2$.
Therefore, given a profile $\vec x = (x_1, x_2, \ldots, x_n)$, the probability of success is the following quantity:
\begin{equation}\label{eq:score}
\psucc(\vec{x}):= \Prob{}{\sum_{i=1}^n w_i \cdot X_i > \frac{1}{2}} + \frac{1}{2}\Prob{}{\sum_{i=1}^n w_i \cdot X_i = \frac{1}{2}} ,
\end{equation}
where $X_i$ is a Bernoulli random variable with  $\Prob{}{X_i = 1} = 1/2 + x_i$.
To give an idea of how the success probability looks like as a function of $\vec{x}$, the first term in the right hand side of \eqref{eq:score} equals
$$\Prob{}{\sum_{i=1}^n w_i \cdot X_i > \frac{1}{2}} = \sum_{S: \sum_j w_j >1/2} \left( \prod_{i\in S: x_i>0} (1/2+x_i)\cdot \prod_{i\not\in S: x_i>0} (1/2-x_i) \right).$$
An equivalent formulation of $\psucc(\vec{x})$, that will be convenient in some sections is     \begin{align}\label{eq:score-rewritten}
\psucc(\vec{x})= \Prob{}{\sum_{i=1}^n Z_i > 0} + \frac{1}{2}\Prob{}{\sum_{i=1}^n Z_i = 0},  
\end{align}
where each $Z_i$ is a random variable with  $\Prob{}{Z_i = x_i} = p_{x_i}:= 1/2 + x_i$ and $\Prob{}{Z_i = -x_i} = p_{-x_i} = 1/2 - x_i = 1- p_{x_i}$. The interpretation when using the $Z_i$ variables, is that a voter $i$ contributes a weight of $x_i$ when voting for the correct outcome, and a weight of $-x_i$ otherwise.

\paragraph{Optimization benchmark.}
    Given a cost function $c(\cdot)$ and budget $B$, the optimization benchmark is defined by maximizing the probability $\psucc(x)$, subject to the following \emph{benchmark budget constraint}:
    \begin{align}\label{eq:budget-benchmark}
        \sum_i c(x_i) \leq B \ . 
    \end{align}
    We denote the corresponding optimum as $OPT(c,B):= \psucc(\vec{x^*})$, where $\vec{x^*}$ is an optimal solution to this problem. \begin{remark}
        Any reward sharing scheme that guarantees nonnegative utilities for the DReps and the budget constraint \eqref{eq:budget_constraint} must satisfy the benchmark budget constraint \eqref{eq:budget-benchmark}.
    \end{remark}

    We will also consider  variants with  ({\it i}) \emph{symmetric} efforts  and ({\it ii}) a \emph{maximum number $k$ of DReps}, 
    \begin{align} 
    x_i \in \{x,0\} && 
        \text{for all $i$ and for some $x$ (symmetry) } \label{eq:symmetry-constraint}\ ;  \\
        |i: x_i>0| \leq k  && 
        (\text{maximum number of DReps}) \ .  \label{eq:dreps-constraint}
    \end{align}

We denote by $OPT^\star(c,B)$ the optimum for the symmetric version \eqref{eq:symmetry-constraint}, by $OPT_k(c,B)$ the optimum for the one with at most $k$ DReps  \eqref{eq:dreps-constraint}, and by $OPT^\star_k(c,B)$ the one in which we have both. 
    By definition, the following relations hold:
    $
        OPT(c,B) \geq OPT^\star(c,B) \geq OPT^\star_k(c,B)$ and  $OPT(c,B) \geq OPT_k(c,B) \geq OPT^\star_k(c,B), 
    $
    which correspond to the optimality of symmetric solutions, and to the optimality of a fixed number of DReps, and combinations thereof.

Finally observe that 
 $x_i \leq x_{\max}(c,B)$ where $x_{\max}(c,B)$ is the largest $x\leq 1/2$ such that $c(x)\leq B$.

\section{Warmup: Equilibrium Analysis of Proportional Sharing}\label{sec:warmup}

We first analyze a very simple and natural reward mechanism where each representative obtains a reward equal to the percentage of the overall delegation that she accumulated, i.e., $ f_i(\vec{w}) =  w_i\cdot B$. Such approaches 
have been considered in many problems in the context of profit sharing games \cite{BachrachSV13}, project games \cite{BiloGM19}, contests \cite{BirmpasKLO22}, etc., due to their simplicity, and usually they also provide good guarantees in terms of performance.

The purpose of this section  is to exhibit that this reward rule is not an appropriate incentive scheme if we are interested in maximizing the probability of the correct outcome, in the sense that it can induce low quality equilibria.

For smooth reward functions $f_i$ and cost function $c$, we can derive the first order conditions that should hold at an equilibrium. If $c'$ is the derivative of $c$, these are 

\begin{equation*}
    \frac{\partial f_i(\vec{x})}{\partial x_i} = c'(x_i),
\end{equation*}

To illustrate our negative result, let us assume that the cost function is linear, $c_i(x_i) = ax_i$. We focus below on the symmetric Nash equilibria that may arise, i.e., profiles where all players exert the same effort.

\begin{theorem}
\label{thm:proportional-nash}
    The only symmetric Nash equilibrium under the proportional reward sharing rule, with a linear cost function, is the strategy profile with $x_i = x = \frac{B(n-1)}{an^2}$ for every $i$, as long as $x\in [0, 1/2]$. 
\end{theorem}

When the parameters $B$ and $a$ are constants independent of $n$, the next theorem shows that for a large enough population of voters, the proportional sharing rule can induce bad equilibria. The reason is that the effort of each DRep is $O(1/n)$ as identified in Theorem \ref{thm:proportional-nash}, and hence the probability of each DRep voting correctly is only $1/2 + O(1/n)$. As a result, the probability of having the correct outcome elected goes to $1/2$ as $n$ becomes large, as established below.

\begin{theorem}\label{thm:pro-sharing-bad}
    Under the symmetric equilibrium profile of the proportional sharing rule, the probability of selecting the right outcome, as $n\rightarrow \infty$, tends to $1/2$. 
\end{theorem}

\begin{remark}
    An analogous conclusion also holds for concave functions of the form $c(x) = x^b$, with $b<1$. There the effort per player at a symmetric equilibrium can be even worse, and bounded by $O((1/n)^{1/b})$.
\end{remark}
\section{Candidate Mechanisms}\label{sec:candmech}

Given the previous conclusions, we suggest that in order to incentivize the delegates to elect the correct outcome, we need a payment rule that induces more competition among them, so as to make an effort to attract delegations. 

\subsection{Equilibria under the Threshold$(k)$ mechanism}\label{sec:eqtopk} 

Consider the mechanism where a delegate receives a reward only if she managed to collect at least a $1/k$-fraction of the total delegation. The reward received by a voter $i$ is:

$$ 
f_i(w_i) = \begin{cases}
			B/k & \text{if } w_i \geq 1/k\\ 
			0 & \text{otherwise}
		\end{cases} \ . 
$$

What kind of equilibria do we expect to have under this mechanism?
Naturally, we cannot have  equilibria with more than $k$ delegates exerting positive efforts, since only up to $k$ delegates can be paid, and the remaining would not have any incentive to make any effort. Instead, we will see that we can have equilibria with exactly $k$ delegates making an effort that are also symmetric as in \eqref{eq:symmetry-constraint}.

The following theorem characterizes the pure Nash equilibria of the Threshold$(k)$ mechanism. In particular, it demonstrates that in every equilibrium the set of players that exerts positive effort is of specific size, while it also provides specific conditions which the efforts of the players must obey.

\begin{theorem}
\label{thm:top-k}
    For the Threshold$(k)$ reward rule, every pure Nash equilibrium must be a symmetric Nash equilibrium in which exactly $k$ voters exert positive effort. Moreover, for any $x\in(0, 1/2]$ there exists an equilibrium where $k$ voters make effort equal to $x$ if and only if one of the following conditions hold:
    \begin{itemize}
        \item either $c(\frac{k}{k-1}x) \geq B/k \geq c(x)$ and $\frac{k}{k-1}x\in (0, 1/2]$,
        \item or $B/k \geq c(x)$ and $\frac{k}{k-1}x >  1/2$, 
    \end{itemize}

\end{theorem}

\begin{proof}[Proof Idea]
The full proof is given in Appendix~\ref{app:proof-thm:top-k}. Here we just describe the main intuition and ideas. 
    First, no equilibrium can have more than $k$ players exerting positive effort (since only the $k$ with the highest effort get rewarded). Also, no equilibrium can have strictly less than $k$ players exerting positive effort (otherwise the player with highest effort can improve her utility by reducing slightly her effort). Hence, any equilibrium must have exactly $k$ players that exert positive effort. Imposing that deviating to zero effort is not profitable, together with $\sum_i w_i=1$, we get that all equilibria are of the form $(x,\dots, x, 0,\dots, 0)$, up to a renaming of the players. The bounds on $c(x)$ are then obtained by considering deviations restricted to equlibria of this form (and  that $w_i \geq 1/k$ is necessary to get the reward $B/k$).  
\end{proof}

There are some positive things that we can claim for this mechanism, showing that there are some advantages against the proportional scheme.
The first one is that in all its equilibria, the DReps who decide to exert a positive effort are making a much higher effort than in the symmetric equilibrium of the proportional scheme identified in Theorem \ref{thm:proportional-nash}. As an example, when $k$ is much smaller than $n$, and under a linear cost function, the effort can be seen to be $\Omega(1/k)$, which is significantly higher than $O(1/n)$ from Theorem \ref{thm:proportional-nash}. Hence we expect to have a higher probability of selecting the correct outcome.

Secondly, the next corollary shows that equilibria always exist and in fact all of them are close to the optimal effort under the constraint of using exactly $k$ DReps. This motivates further the question of identifying the optimal $k$ for maximizing the success probability for electing the correct outcome, which is the focus of Section \ref{sec:optsol}. 

\begin{corollary}
\label{cor:top-k-opt}
    Let $x^*(k)$ be the optimal effort for the optimization benchmark in which we impose a committee of size exactly equal to $k$. 
    \begin{itemize}
        \item The profile where $k$ people exert effort equal to $x=x^*(k)$ is an equilibrium.
        \item For any equilibrium profile with effort $x$, such that $\frac{k}{k-1}x\leq 1/2$, it holds that $x\geq (1-\frac{1}{k})x^*(k)$.
    \end{itemize}
\end{corollary}

\begin{proof}
    When requiring a committee of exactly $k$ DReps, the optimal effort has to satisfy that $kc(x^*(k)) = B$.
    Hence $B/k = c(x^*(k))$. This means that for $x=x^*(k)$, the equilibrium conditions of Theorem \ref{thm:top-k} are satisfied.
    For the second part of the corollary, note that for any other equilibrium profile with effort $x$, and with $\frac{k}{k-1}x\leq 1/2$, by Theorem \ref{thm:top-k}, we have that $c(\frac{k}{k-1}x) \geq B/k = c(x^*(k)) $. Since our function is increasing, it should hold that $\frac{k}{k-1}x \geq x^*(k)$, and thus $x\geq (1-\frac{1}{k})x^*(k)$.
    \end{proof}

\subsection{Variants of Threshold$(k)$}\label{sec:var}

There are two variations of Threshold$(k)$ that are also of interest. Both of them require that we spend all the budget in contrast to the rule we have considered where we may end up spending less than the available budget. In the first variant below, the budget is split up to exhaustion, among the DReps who collected delegations that are at least a fraction of $1/k$:

\begin{align}
	f_i(w_i) = \begin{cases}
		\frac{B}{|j: w_j\geq 1/k|} & \text{if } w_i \geq 1/k\\ 
		0 & \text{otherwise}
	\end{cases}
\tag{Variant 1}
\end{align}

\begin{theorem}
\label{thm:v1}
For any continuous and strictly increasing cost function c, the mechanism of Variant 1 does not possess pure Nash equilibria.
\end{theorem}

Consider now a different variation, where again up to $k$ DReps may receive a reward but the budget is then split proportionately among them:
\begin{align}
	f_i(w_i) = \begin{cases}
		B\cdot \frac{w_i}{\sum_{j: w_j\geq 1/k} w_j} & \text{if } w_i \geq 1/k\\ 
		0 & \text{otherwise}
	\end{cases}
\tag{Variant 2}
\end{align}

\begin{remark}
The reasoning in the the proof of theorem \ref{thm:v1} implies that this variant also does not have pure Nash equilibria that are symmetric among the delegates who exert some positive effort. We point out however, that, in contrast to Variant 1, there might be non-symmetric equilibria or equilibria with less than $k$ people making positive effort. The reason for this is that the reward of a player $i$ when $w_i \geq 1/k$, is now affected by the effort that she exerts. The existence of such equilibria highly depends on the cost function $c(\cdot)$, something that (along with their more complex nature) makes this variant less appealing than Threshold$(k)$.

\end{remark}

\section{Optimal Solutions}\label{sec:optsol}

Motivated by the previous results on proportional sharing and the subsequent Threshold$(k)$ mechanism and Corollary \ref{cor:top-k-opt}, we investigate now the problem from an optimization perspective. Hence, we focus on determining the optimal number $k$ of DReps for a given budget. That is, the problem of maximizing the probability of success in the symmetric case, where the effort exerted is the same for all the DReps who do so as described by \eqref{eq:symmetry-constraint} (which in some cases turns out to be the optimal even under the more general solution space where the voters could exert different efforts).

\subsection{General Bounds}\label{sec:genboun}

We start with proving a general upper bound on the probability of success which applies to any number of DReps and to the \emph{non-symmetric} case (Theorem~\ref{thm:psucc-ub-general}). This result plays a key role for the analysis of concave costs (Section~\ref{sec:concave}) as well as for the more general concave-convex ones (Section~\ref{sec:concave-convex}).  By a slight abuse of notation, in the sequel, for effort vectors of the form $\vec{x} = (x_1,0,\ldots, 0)$, we will use $\psucc(x_1)$ instead of $\psucc(\vec{x})$, and similarly for $\psucc(x_1,x_2,\ldots,x_n)$.

\begin{theorem}\label{thm:psucc-ub-general}
	For any number of DReps with (possibly 
 non-symmetric) efforts $\vec{x}= (x_1,x_2,\ldots,x_n)$, it holds that 
	\begin{align}
		\psucc(\vec{x})
  \leq \psucc(y) = 1/2 + y\ ,  && y = x_1+x_2+\cdots+x_n \ .
	\end{align}
\end{theorem}

Note that the above result implies that a single DRep is optimal whenever effort $y$ does not violate the budget constraint, i.e., if $c(y)\leq B$. The dependency on the budget $B$ is somewhat unavoidable, even for the case of convex costs (Section~\ref{sec:convex}).

\subsection{Concave (and linear) case: One DRep is optimal}\label{sec:concave}

An immediate consequence of Theorem~\ref{thm:psucc-ub-general} is that for concave (and thus for linear) costs, the optimal solution consists of a single DRep. 
Indeed, for any concave cost function $c(\cdot)$, and any effort vector $(x_1,\ldots, x_n)$, the solution which utilizes  a single DRep with effort $y=x_1+\cdots+ x_n$ still satisfies the budget constraint \eqref{eq:budget-benchmark}, thus implying the following.
 \begin{corollary}\label{cor_concave}
     For any concave cost function, there is always an optimal solution consisting of a single DRep with effort $x_1$ such that $c(x_1)=B$, where $B$ is the budget. Hence, using asymmetric efforts does not help in this case. Moreover, for the strictly concave case, one DRep is strictly better than several DReps. 
 \end{corollary}

\subsection{Concave-convex case: Too many DReps are not optimal}\label{sec:concave-convex}

In this section, we consider the more general class of cost functions, namely, concave-convex ones, which arise in many real world  settings as a means to capture more accurately the cost of information acquisition (in Section~\ref{sec:S-shape} we elaborate on this). An example showing the shape of concave-convex functions can be seen in Figure \ref{fig:geometric}.
We shall prove below general bounds on the maximum number of DReps that produce optimal solutions, and the corresponding  efforts. The main message of this section is that the optimal number of DReps cannot be too high, and can be upper bounded by appropriate parameters with a geometric interpretation that we define below (and hence also establishing lower bounds on the minimal effort at an optimal solution).

For $x>0$, the maximum number of DReps with equal effort $x$  that we can use, given a cost function $c(\cdot)$ and budget $B$, is equal to
 $ \DRepsopt{c}{x} := \left\lfloor \frac{B}{c(x)} \right\rfloor\ . $

The following lemma states some natural properties of the success probability. 
\begin{lemma}\label{le:prob-succ-mon}
	The success probability $\psuccun(x,k)$  is monotone increasing in both the effort $x$ and in the number of DReps $k$. That is, $\psuccun(x,k)\leq \psuccun(x,k+1)$ and $\psuccun(x,k)< \psuccun(x',k)$ for all $x'>x$.  
\end{lemma}


\begin{figure}
	\centering
	\includegraphics[scale=.8]{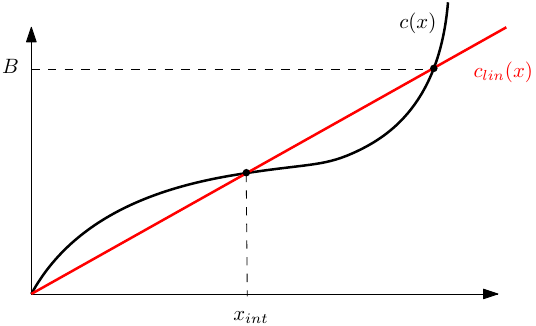}
	\caption{The idea of Theorem~\ref{thm:concave-convex-lb-effort}.}
	\label{fig:geometric}
\end{figure}

Our first result (Theorem~\ref{thm:concave-convex-lb-effort} below) is based on a geometric argument shown in Figure~\ref{fig:geometric}, and formally captured by this definition. 

\begin{definition}\label{def:concave-convex-line-general}
For any concave-convex cost function $c(\cdot)$ and for any budget $B$,
    we let $c_{lin}(x) = \alpha_{c,B}\cdot x$ be the linear cost function such that $c(\cdot)$ and $c_{lin}(\cdot)$  take the same value $B$ at some common point $x_1>0$, i.e., $c(x_1)=B=c_{lin}(x_1)$. Then, we denote by $x_{int} = x_{int}(c,B)$ the largest value $x$ such that $c_{lin}(x) \leq c(x)$ for all $x\in [0,x_{int}(c,B)]$ such that $c(x)\leq B$. 
\end{definition}

We first provide a general lower bound on the minimum effort  for the optimal symmetric solutions, thus implying that the number of DReps cannot be too large. 

\begin{theorem}\label{thm:concave-convex-lb-effort}
	For any concave-convex cost function $c(\cdot)$ and any budget $B$, the  optimal symmetric solution with identical efforts must use at least an effort level of $x_{int}$, and thus at most $\DRepsopt{c}{x_{int}}$ DReps, where $x_{int}=x_{int}(c, B)>0$ is given in Definition~\ref{def:concave-convex-line-general}.  
\end{theorem}

\begin{remark}
    Note that for smaller values of the budget $B$, the bounds provided by the previous theorem get better, as the critical value $x_{int}=x_{int}(c,B)$ increases (see Figure~\ref{fig:geometric}), which rules out more values from the optimum. Conversely, for increasing values of the budget $B$, the bounds of the previous theorem become weaker, as the opposite happens.  
\end{remark}

Moving on, the influence of the budget $B$ suggested by the previous remark is the main focus of the remaining of this (and the next) section. We first note that Theorem~\ref{thm:concave-convex-lb-effort} implies a \emph{stronger} version of the result for concave costs. 
The idea  is shown in  Figure~\ref{fig:geometric-tangent}, and formalized by the next definition.

\begin{definition}\label{def:points}
	For any concave-convex cost function $c(\cdot)$, let $x_{inflection}^c$ denote its inflection point (where the switch from concave to convex occurs). Moreover, let  $x_{tangent}^c$ be the largest  point  such that the line from the origin passing through $x_{tangent}^c$ is entirely below the cost function, and the two curves intersect at $x_{tangent}^c$. That is, the linear cost $c_{lin}(\cdot)$ such that $c_{lin}(x_{tangent}^c)= c(x_{tangent}^c)$ satisfies $c_{lin}(x) \leq c(x)$ for all $x$ such that $c(x)\leq B$.   
\end{definition}

\begin{figure}
	\centering
	\includegraphics[scale=.8]{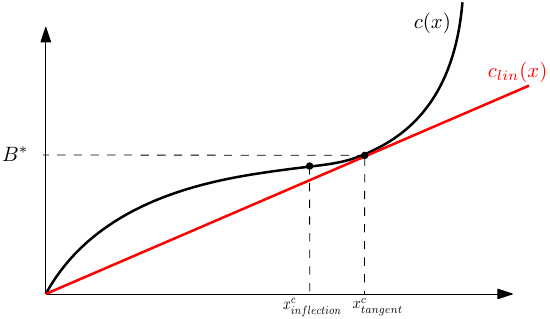}
	\caption{The idea of Theorem~\ref{thm:cocave-convex-tangent} with the inflection point $x_{inflection}^c$ and point $x_{tangent}^c$.}
	\label{fig:geometric-tangent}
\end{figure}

Observe that for $B= B^*:= c(x_{tangent}^c)$, the point  $x_{int}(c,B)$ in Theorem~\ref{thm:concave-convex-lb-effort} coincides with $x_{tangent}^c$, and $\DRepsopt{c}{x_{int}}=1$; the same is true for $B< B^*$ as we  consider the cost function restricted to a smaller interval -- condition $c(x)\leq B$ in Definitions~\ref{def:concave-convex-line-general} and \ref{def:points}. 

\begin{corollary}\label{cor:geometric}
For any concave-convex cost function $c(\cdot)$, and for any budget $B\leq B^*$, there is always an optimal solution consisting of a single DRep with effort $x_1$ such that $c(x_1)=B$, where  $B^*=c(x_{tangent}^c)$.  
\end{corollary}

Note that for the concave case the condition on the budget 
required by the previous corollary is always satisfied, and therefore Corollary~\ref{cor:geometric} generalizes Corollary~\ref{cor_concave}. 

We conclude this section by considering an arbitrarily large budget $B$ and an improved version of the result in Theorem~\ref{thm:concave-convex-lb-effort}, if the following condition holds. 

\begin{assumption}[monotonicity]\label{as:prob-succ-lin-mon}
	For any linear cost function $c_{lin}(x) = a\cdot x$ with $a>0$, the success probability is monotone decreasing with the number of DReps (while using the same corresponding maximum  effort per DRep). That is, for any $x \leq x'$, 
	$
		\psucc(x,k) \leq \psucc(x',k') , 
	$
where $k = \DRepsopt{c_{lin}}{x}$ and $k' = \DRepsopt{c_{lin}}{x'}$.
\end{assumption}

The next theorem provides better bounds for large $B$, showing that the optimal solution with equal efforts is situated in the region where $x_{tangent}^c > x_{int}(c, B)$. This result is conditioned on the above assumption, which we verified experimentally for several values of $a$ and $k$, although we have not been able to formally prove it. In Appendix~\ref{app:assumption-weaker} we provide further evidence for this assumption by proving a slightly weaker version of it. The idea  of the proof of the next theorem is shown in  Figure~\ref{fig:geometric-tangent}.

\begin{theorem}\label{thm:cocave-convex-tangent}
	Suppose Assumption~\ref{as:prob-succ-lin-mon} holds. Then, for any concave-convex cost function $c(\cdot)$, and any budget $B$, the optimal solution with identical efforts must use at least $x_{tangent}^c$ effort level and thus at most $\DRepsopt{c}{x_{tangent}^c}$ DReps. 
\end{theorem}

\subsection{Convex case: The budget matters for the optimal number of DReps}\label{sec:convex}

One might conjecture that the convex case behaves inversely to the concave case, implying that many DReps are superior to fewer. However, we demonstrate that this presumption does not hold true, revealing a more intricate scenario (shown in Theorem~\ref{th_three-vs-one} below). In particular, it holds that:
\begin{enumerate}
    \item For certain convex costs, the superiority of three DReps over one hinges on the budget $B$.
    \item The resolution of the former inquiry is contingent upon the specific convex function, even within  costs of the form $c(x)=x^\beta$, parameterized in  $\beta>1$.
\end{enumerate}
This indicates that a straightforward ``monotonicity" argument asserting the supremacy of larger committees does not universally apply in the convex case.

The next lemma allows us to compare $k=3$ with $k=1$  on general cost functions, and it will also be used below in Section~\ref{sec:S-shape} for a specific class of concave-convex cost functions. 

\begin{lemma}\label{le:three-vs-one}
Given any cost function $c(x)$ and any budget $B£$,  three DReps are better than one DRep (w.r.t. the success probability) if and only if the corresponding optimal efforts, $x^*(1)$ and $x^*(3)$, satisfy
\begin{align}\label{eq:three-vs-one}
   x^*(1)<\frac{3x^*(3) - 4 x^*(3)^3}{2}  \ . 
\end{align}
\end{lemma}

The above lemma yields the following result for convex costs. We note that the result is analogous to the one obtained in \cite{gersbach2024we}, which also considers convex costs though in a slightly different setting. 

\begin{theorem}\label{th_three-vs-one}
    For the family of convex cost functions $c(x) = x^\beta$, with $\beta>1$, and for $\beta^\star := \frac{\ln(3)}{\ln(3)-\ln(2)}\approx 2.7095$, the following holds: 
    \begin{enumerate}
        \item For $\beta < \beta^\star$, one DRep is always better than three DReps, regardless of the budget $B$.
        \item For $\beta > \beta^\star$, three DReps with equal effort are better than a single DRep if and only if the budget $B$ is at most $B^\star=3^{3/2} \cdot \left(\frac{3^{1 - 1/\beta}  - 2}{4}\right)^{\beta/2}$.
    \end{enumerate}
\end{theorem}

We next provide a numerical example for a simple (low degree) convex cost function. 

\begin{example}[budget-dependence for convex functions]
    Consider the convex cost function $c(x) = x^4$. By applying Theorem \ref{th_three-vs-one}, we have that three DReps with equal effort are better than a single DRep if and only if the budget $B$ is at most $B^\star =  \left(\frac{3^\frac{3}{2}-2{\cdot}3^\frac{3}{4}}{4}\right)^2\approx 0.0254$. Moreover, the largest benefit for three DReps against a single DRep is for a budget equal to  $B^{\star \star} = \left(\frac{4\sqrt[4]{3}\,x^2-3^\frac{3}{4}+2}{2}\right)^4$.
\end{example}

Overall, the question of determining the optimal committee size with convex functions seems quite challenging. Even generalizing Theorem \ref{th_three-vs-one}, to understand e.g. if $k''$ DReps are better than $k'$ DReps for $k'' > k' >3$, leads to a more complex analysis, since one needs to account for all the different probability events that can lead to the correct outcome. We leave this as an interesting open problem for future work.

\section{Application to S-shaped learning curves}\label{sec:S-shape}
\renewcommand{\sigma}{\xi}

In this section, we focus on certain classes of concave-convex costs which are derived by some \emph{S-shaped learning curves}  considered in the literature of experimental psychology
\cite{murre2014s}. Specifically, \cite{murre2014s} proposes the following family of exponential learning curves, for learning over time a new task or new material (in our case, learning the correct outcome),
\begin{align}
    \label{eq:learning-S-shape-exp}
    p_\xi(t) = [1 - \exp(-\mu t)]^\xi
\end{align}
where the above success probability depends on the following parameters: 
\begin{enumerate}
\item $t$ is the time spent on learning (cost);
    \item $\mu$ is the learning rate (potentially different for each individual);
    \item $\xi$ is a complexity parameter (equal for all individuals).
\end{enumerate}

The complexity parameter $\xi>1$ is crucial in order to obtain S-shaped learning curves, which have been often observed in practice. As pointed out by \cite{murre2014s}, several prior theoretical models \cite{Estes1950-ESTTAS-2,hull1943principles,rescorla1972theory} assume a ``vanilla'' exponential function with $\xi=1$, i.e.,  
$p(t) = 1 - \exp(-\mu t)$. The major shortcoming of these theoretical models is the mismatch with the S-shape observed by experimentalists. Intuitively speaking, the refined model in \cite{murre2014s} assumes that $\xi=1$ corresponds to acquire some ``elementary skills'', while the actual tasks to be solved require some complex \emph{combination} of these skills. The latter is captured by a parameter $\xi>1$.   

\begin{remark}[homogeneous $\mu$]\label{rem:exp-learning-method}
    Though the above model starts with the assumption of a different learning rate per individual, the actual experiments in \cite{murre2014s} are conducted by isolating groups of people with \emph{similar learning rate} and then fit the data of the group in order to estimate the complexity parameter $\xi$ (the same $\xi$ is used across all groups). In the following, we shall consider  $\mu$ as homogeneous across all individuals, which is in line with the methodology in  \cite{murre2014s} just described. 
\end{remark}

We next derive a concave-convex cost function $c(x)$ which corresponds to the above families of functions \eqref{eq:learning-S-shape-exp}. 
In our notation,  $t = c(x)$. Moreover, we consider $x=p_\xi(t)/2$ so that the  success probability for our two-outcomes  setting, $1/2 + x = 1/2 + p_\xi(t)/2$,  attains its minimum of $1/2$ when the cost for the corresponding player is zero. We thus have this relation: 
$2x = p_\xi(t) = [1 - \exp(-\mu c(x))]^\xi$, 
which is equivalent to
$\exp(-\mu c(x)) =  1 - (2x)^{1/\xi}\ . $
The latter identity leads to the next definition for $c(x)$.

\begin{definition}\label{def:exp-learning-costs}
   The  cost function associated to the exponential-learning curve  
   \eqref{eq:learning-S-shape-exp}  with complexity $\xi$ and learning rate $\mu$ is defined as 
\begin{align}
    \label{eq:cost-S-shaped-exp}
    c_{\xi}(x) :=  -\frac{1}{\mu}\ln \left(1 - (2x)^{1/\xi}\right)\ . 
\end{align}
\end{definition}

\begin{figure}
    \centering
    \includegraphics[scale=.4]{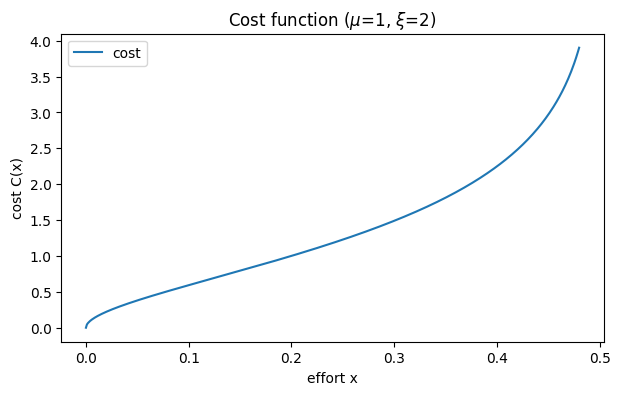}
    \caption{An example of a concave-convex cost function from exponential learning (Definition~\ref{def:exp-learning-costs}).}
    \label{fig:exp-learning-2}
\end{figure}

Figure~\ref{fig:exp-learning-2} shows an example of these concave-convex cost functions. 
The next lemma provides useful features of these cost functions.

\begin{lemma}\label{le:inflection-S-shape}
For every complexity parameter $\xi>1$, 
    the cost function $c_{\xi}$ in \eqref{eq:cost-S-shaped-exp} is a concave-convex function whose inflection point does not depend on the (individual) learning parameter $\mu$ but only on $\xi$. In particular, the inflection point is $\Tilde{x}= \frac{1}{2}(1 - 1/\xi)^\xi$ and the corresponding cost is $c_{\xi}(\Tilde{x})= \ln(\xi)/\mu$.
\end{lemma}

We next quantify the optimal symmetric efforts for this family of functions.

\begin{lemma}\label{le:opt-efforts-S-shape}
    For the cost function $c_{\xi}$  in \eqref{eq:cost-S-shaped-exp} the optimal symmetric effort $x^*(k)$, given a budget $B$, is equal to 
    $
        x^*(k) = \frac{p_\xi(B/k)}{2} = \frac{[1-\exp(-B\mu/k)]^\xi}{2} \ . 
    $
\end{lemma}

The next result concerns the optimality of one DRep when the budget is ``not too high''. Note that the result applies to a range of values for the budget, for which the largest feasible effort can be bigger than the inflection point, and thus we are still effectively considering a concave-convex cost function. 
\begin{corollary}\label{cor:S-shape-one-vs-three}
    For the cost function $c_{\xi}$  in \eqref{eq:cost-S-shaped-exp} one DRep is optimal for any budget $B \leq B^*$, where  $B^*$ is defined as in Corollary~\ref{cor:geometric} and it satisfies $B^*>\ln(\xi)/\mu$. In this case, the optimal probability of success is equal to $1/2 + p_{\xi}(B)/2 = 1/2 + \frac{[1-\exp(-B\mu)]^\xi}{2}$.  
\end{corollary}

Similarly to the convex case in Section~\ref{sec:convex}, one DRep \emph{is not} always optimal, and this depends on the budget $B$. In particular,  Lemma~\ref{le:three-vs-one} leads to the following observation.

\begin{observation}
    There exists a value $B_{three}$ such that, for any budget $B\geq B_{three}$ three  DReps are better than one DRep. This is because, for sufficiently large $B$, we are in the region where the curve gets sufficiently steep (see Figure~\ref{fig:exp-learning-2}), and  condition \eqref{eq:three-vs-one} of Lemma~\ref{le:three-vs-one} is satisfied for the values given by Lemma~\ref{le:opt-efforts-S-shape}. For example, this happens at $B_{three}\approx 4.35$ for the curve in Figure~\ref{fig:exp-learning-2} ($\mu=1$ and $\xi=2$).  
\end{observation}

We conclude this section, by  considering the equilibrium conditions for the Threshold$(k)$ mechanism. In particular, Theorem~\ref{thm:top-k} implies the following result. 

\begin{theorem}\label{th:S-shape-euilibria}
    For the cost function $c_{\xi}$  in \eqref{eq:cost-S-shaped-exp}, and for budget $B$, every symmetric equilibrium of the Threshold$(k)$ mechanism has exactly $k$ DReps with positive effort $x$, and  $x$ satisfies the following condition:
    \begin{itemize}
        \item either $\frac{k-1}{k} \cdot x^*(k) \leq  x \leq  x^*(k)$ and $\frac{k}{k-1}x \in (0,1/2]$,
        \item or $x \leq  x^*(k)$ and $\frac{k}{k-1}x >1/2$. 
    \end{itemize}
Moreover,  $x^*(k) = \frac{p_\xi(B/k)}{2} = \frac{[1-\exp(-B\mu/k)]^\xi}{2}$ is the optimal effort for symmetric solutions with $k$ DReps and budget $B$.
\end{theorem}

\section{Conclusions}\label{sec:con}

We have explored the questions of designing reward schemes and determining an appropriate number of representatives both from a game-theoretic and an optimization viewpoint. For proposing reward schemes, our results reveal that threshold-like mechanisms are more preferred as they seem to incentivize better the DReps on exerting more effort. Regarding the question of determining optimal committee sizes, our findings differ based on the type of cost function and on the available budget. In many cases however, as revealed in Section \ref{sec:optsol}, a small number of representatives seems appropriate for increasing the chances to select the correct outcome in the underlying election.

\bibliographystyle{plain}
\bibliography{biblio}

 \newpage
 \appendix
\section{Appendix with Omitted Proofs}

\subsection{Proof of Theorem~\ref{thm:proportional-nash}}
\begin{proof}
	
	Given a strategy profile $\vec{x} = (x_1,\dots, x_n)$, the utility of a player $i$ is
	$$\frac{x_i\cdot B}{\sum_j x_j}  - ax_i\ . $$
	By taking the first order conditions we get that at equilibrium, we must have:
	$$\frac{\partial f_i(\vec{x})}{\partial x_i} = c'(x_i) \Leftrightarrow B\cdot \frac{\sum_{j\neq i} x_j}{(\sum_j x_j)^2} = a ~~\forall i\ . $$
	Since we are looking for symmetric equilibria, we can set $x_j=x$ for every $j$ and solve the above system of equations. Then the nominator above becomes $B(n-1)x$ and the denominator equals $n^2x^2$, and therefore, the solution we get for $x$ is precisely the quantity in the statement of the theorem.
\end{proof}

\subsection{Proof of Theorem~\ref{thm:pro-sharing-bad}}
\begin{proof}
Without loss of generality, assume for convenience that $B=1$ and that $n$ is odd. Consider the symmetric equilibrium profile $\vec{x}$, with $x_i=x= \frac{n-1}{an^2}$ for every $i$.  Then, the probability of selecting the correct outcome equals

\begin{equation}
    \label{eq:correct}
    \psucc(\vec{x}) = \sum_{i=n/2+1}^n {n \choose i} (1/2+x)^i (1/2-x)^{n-i} \ . 
\end{equation}
After substituting the value of $x$, we get
\begin{align*}  
\psucc(\vec{x}) & = \sum_{i=n/2+1}^n {n \choose i} (1/2+ \frac{(n-1)}{an^2})^i (1/2-\frac{(n-1)}{an^2})^{n-i}  \\
& = \left( \frac{1}{2an^2} \right)^n \cdot \sum_{i=n/2+1}^n {n \choose i} (an^2+2n -2)^i (an^2-2n +2)^{n-i} \ . 
\end{align*}
Note now that for every $i$ in the summation, the term $(an^2+2n -2)^i (an^2-2n +2)^{n-i}$ is asymptotically equal to $(an^2)^n + o((an^2)^n)$.
By taking this out of the sum, and taking the limit, we have after the cancellations of these terms that:
$$\lim_{n \rightarrow \infty} \psucc(\vec{x}) = \lim_{n \rightarrow \infty} \frac{\sum_{i=n/2+1}^n {n \choose i}}{2^n} = \frac{1}{2}$$
where we used the fact that $\sum_{i=n/2+1}^n {n \choose i} = 2^{n-1}$.
\end{proof}

\subsection{Proof of Lemma~\ref{le:prob-succ-mon}}
\begin{proof}
    For the first part, we argue as follows. For any good subset $S$ for $k$ DReps, i.e., any subset of at least $k/2$ DReps, we construct a new subset $\hat S:=S \cup\{k+1\}$ which is good for $k+1$ DReps. Note that distinct subsets $S\neq S'$ for $k$ DReps are mapped into different new subsets $\hat{S}\neq \hat{S'}$ as we are adding the new DRep $k+1$. We now formally prove that the success probability in the new subsets is higher. Indeed,  
the initial probability is 
\begin{align}\label{eq:succ-subset}
    P_S(x):= (1/2+x)^{|S|}(1/2-x)^{k-|S|} 
\end{align}
and the probability of the new subset is
\begin{align}\label{eq:mapping-prob}
    P_{\hat{S}}(x): = (1/2+x)^{|S|+1}(1/2-x)^{k-|S|-1} =P_S(x) \cdot \frac{1/2+x}{1/2-x} \geq  P_S(x) \ . 
\end{align}
Putting things together, 
\begin{align*}
    \psucc(x,k) = & \sum_{S: |S|>k/2} P_S(x) + \frac{1}{2}\sum_{S: |S|=k/2} P_S(x) \\
    \stackrel{\eqref{eq:mapping-prob}}{<} & \sum_{S: |S|>k/2} P_{\hat{S}}(x) + \frac{1}{2}\sum_{S: |S|=k/2} P_{\hat{S}}(x)
    \\
    \leq & \sum_{\hat{S}: |\hat{S}|>(k+1)/2} P_{\hat{S}}(x) + \frac{1}{2}\sum_{\hat{S}: |\hat{S}|=(k+1)/2} P_{\hat{S}}(x) = \psucc(x,k+1) \ , 
\end{align*}
where the last inequality follows from the observation that, in our mapping, $|\hat{S}|=|S|+1$ and therefore $|\hat{S}|\geq k/2 + 1>(k+1)/2$ whenever $|S|\geq k/2$.

As for the second part, we show that the  probability of failure associated to each bad subset $S$, i.e., a subset of at most $k/2$ DReps,  is monotone decreasing in $x$. Indeed, for any $S$ with $|S|\leq k/2$, 
 \begin{align*}
    P_S(x) \stackrel{\eqref{eq:succ-subset}}{=} & ((1/2+x)(1/2-x))^{|S|}(1/2+x)^{k-2|S|} \\
    & =  (1/4-x^2)^{|S|}(1/2-x)^{k-2|S|} 
    \ , 
\end{align*}
and  because $k-2|S|\geq 0$ both terms are decreasing in $x$. Therefore $P_S(x)\geq P_S(x')$ for any $x'>x$, thus implying the following inequality:
\begin{align*}
    \pfail(x,k) := 1 - \psucc(x,k) = & \sum_{S: |S|<k/2} P_S(x) + \frac{1}{2}\sum_{S: |S|=k/2} P_S(x) \\
    \geq & \sum_{S: |S|<k/2} P_S(x') + \frac{1}{2}\sum_{S: |S|=k/2} P_S(x')
    \\
    = & \pfail(x',k)  = 1 - \psucc(x',k)
\end{align*}
which completes the proof. 
\end{proof}

\subsection{Proof of Theorem~\ref{thm:top-k}}\label{app:proof-thm:top-k}
\begin{proof}
	First, notice that by the definition of the mechanism, if more than $k$ players exert positive effort, only $k$ of them (the ones with the highest effort) will be rewarded. To see this, observe that given the fact that  $\sum_i w_i=1$, it is not possible for more than $k$ in total $w_i$'s to be bigger than $1/k$. Therefore, there can be no equilibrium where more than $k$ players exert a positive effort, since then the lowest-effort players would have an incentive to deviate to zero effort. At the same time, there can be no equilibrium where the set of players with positive effort has a cardinality that is less than $k$. To see this, assume for contradiction that there exists an equilibrium where $S$ is the set of players that exerts positive effort, and and for which we have that $|S|=l<k$. For each player $i \in S$, it holds that $w_i \geq 1/k$ as otherwise $i$ could deviate to zero effort and improve her utility. Now, let $i^*$ be the player in $S$ that exerts the maximum effort. It is easy to see that $\frac{x_{i^*}}{\sum_{i \in S}x_i} \geq 1/l>1/k$. By the latter, we derive that there always exists an $\epsilon>0$ such that $\frac{x_{i^*}-\epsilon}{\sum_{i \in S}x_i-\epsilon} \geq 1/k$. Therefore, player $i^*$ could deviate to $x_i^* -\epsilon$ effort and claim the reward at a lower cost (recall that the cost function is strictly increasing). As this deviation improves her utility, we end up to a contradiction. 
	
	From the above discussion, we get that at an equilibrium we have exactly $k$ players that exert positive effort, and for every such player $i$, we have $w_i \geq 1/k$ (as otherwise she could improve her utility by deviating to zero effort). Combining the aforementioned conditions with the fact that $\sum_i w_i=1$, we get that all the players with $x_i>0$ exert exactly the same effort. Therefore, we conclude that the only profiles that we need to examine are the ones where exactly $k$ players make a positive and equal effort.

	Consider such a strategy profile in the form $(x,\dots, x, 0,\dots, 0)$.
	Let us look at a player who does not make any effort. If this player deviates in order to get better off, the only meaningful action is to select an effort level $\bar{x}>x$ so as to attract a delegation of at least $1/k$. To do so, $\bar{x}$ should satisfy 
	$$ \frac{\bar{x}}{kx+ \bar{x}} \geq 1/k ~~\Rightarrow~~ \bar{x} \geq \frac{k}{k-1}x \ . $$
	If $\frac{k}{k-1}x>1/2$, then this is not a feasible deviation. Otherwise, the deviation is feasible, but since we do not want it to be a successful one, the total utility after the deviation should be non-positive. Therefore we must have that $c(\bar{x}) \geq B/k$ for any $\bar{x} \geq \frac{k}{k-1}x$. The cost function is non-decreasing, hence it suffices that 
	$$c\left(\frac{k}{k-1}x\right) \geq \frac{B}{k} \ . $$
	
	We come now to the players who exert effort $x$.
	None of these players has an incentive to exert a higher effort, since they will not receive a higher payment. If on the other hand such a player $i$ makes a lower effort $x'<x$, then her reward drops to $0$, as $\frac{x'}{\sum_{j \neq i}x+x'}<\frac{x}{\sum_{j}x}=\frac{1}{k}$. Hence, it suffices that her initial utility is non-negative to ensure we are at an equilibrium, i.e., $B/k\geq c(x)$. This completes the proof.
\end{proof}

\subsection{Proof of Theorem~\ref{thm:v1}}
\begin{proof}
	As in the proof of Theorem \ref{thm:top-k}, there can be no pure equilibrium where more than $k$ players exert a positive effort.
	
	Hence, let us focus on the existence of equilibria that are symmetric w.r.t. the DReps who exert a positive effort.
	Consider first a profile where $k$ people exert a positive effort $x$. 
	Suppose now that one of these $k$ DReps is deviating to exert effort $x+\epsilon$, for some small $\epsilon>0$. 
	By doing so, she will become the only player attracting at least a fraction of $1/k$ of delegations (all the others will attract slightly less than $1/k$ of the total delegations). Hence she will have the budget $B$ by herself. 
	Now for this to be a successful deviation, the following needs to hold for some $\epsilon$:
	$$ B - c(x+\epsilon) > \frac{B}{k} -c(x) \Rightarrow c(x+\epsilon) - c(x) < (1-\frac{1}{k})B $$
	Since the cost function is continuous, we can always find an $\epsilon$ satisfying the above condition for any $x$. Therefore, the profile we started with cannot be an equilibrium.  
	
	Consider now a profile where $\ell$ people exert a positive effort $x$, for some $\ell< k$. 
	Suppose now that one of these $\ell$ DReps, say voter $i$, is deviating to exert effort $x-\epsilon$. It is straightforward that since $\ell<k$, there exists a small enough $\epsilon$ so that  $x-\epsilon$ is still at least a $1/k$ fraction of the total effort exerted. Hence DRep $i$ will still attract at least $1/k$ of the total delegations, and the same is true for the other $\ell-1$ DReps. This means that $i$ will continue to receive the same reward but with a lower cost since the cost function is strictly increasing, which means that the profile cannot be an equilibrium.

	Finally, consider a non-symmetric profile, where at least 2 DReps have different efforts, say $x_i > x_j$, where $x_i$ is the maximum effort exerted. But then DRep $i$ can lower her effort by some $\epsilon$ and still maintain at least a $1/k$ fraction of the delegations, and the same reward. Hence such a profile also cannot be an equilibrium.
\end{proof}

\subsection{Proof of Lemma~\ref{le:three-vs-one}}
\begin{proof}
    For $x=x^*(1)$ and $y=x^*(3)$, we consider the function
    \begin{align}
   f(x,y) := &   \psucc(x^*(1)) - \psucc(x^*(3),x^*(3),x^*(3))  \label{eq:fun-1-vs-3-gen}
    \\ 
    =&  1/2 + x - 1 + (1-p_{x^*(3)})^3 + 3p_{x^*(3)}(1-p_{x^*(3)})^2 \\ = & x-\dfrac{1}{2} +\left(\dfrac{1}{2}-y\right)^3+3\left(\dfrac{1}{2}-y\right)^2\left(\dfrac{1}{2}+y\right)
\end{align}
whose roots are given by the identity
$$x=\frac{y(3 - 4 y^2)}{2} \ , $$
for $y < 1/2$ which implies  $x>0$. 
The lemma follows from the observation that $f(x,y)$ is increasing in $x$ and, by definition, $f(x,y)<0$ if and only if $\psucc(x^*(1)) < \psucc(x^*(3),x^*(3),x^*(3))$. \end{proof}

\subsection{Proof of Theorem~\ref{thm:psucc-ub-general}}\label{sec:proof:thm:psucc-ub-general}

The remaining of this subsection is dedicated to the proof of Theorem~\ref{thm:psucc-ub-general}.
In order to facilitate the recursion that we use in the proof below, we introduce an additional definition on the 
probability that a DRep votes correctly, and also a variant on the overall success probability. First, we let $C(y,p_y^C)$ denote an artificial ``compound'' DRep (to be used shortly), who contributes a positive weight $y$ on the overall success with some probability $p_y^C$, and a negative weight $-y$ with probability $p_{-y}^C = 1 - p_y^C$. Second, we let $\psuccub(\cdot)$ be the probability of success for the variant in which, when the votes lead to a tie, the correct outcome is selected with probability $1$, instead of $1/2$. We thus have 
\begin{align}\label{eq:p_succ-ub}
	\psucc(x_1,x_2,\ldots,x_n) \leq \psuccub(x_1,x_2,\ldots,x_n)\ . 
\end{align}
Without loss of generality, assume $x_1\geq x_2$ and let $S$ denote the event that DReps $1$ and $2$ vote in the same way, and by $\overline{S}$ its complement. Recall also that $p_{x_i} = 1/2+x_i$. We then have
\begin{align}
	P(S) = p_{x_1}p_{x_2} + (1-p_{x_1})(1-p_{x_2})\ ,  && P(\overline{S}) = p_{x_1}(1-p_{x_2}) + (1-p_{x_1})p_{x_2} = 1-P(S) \ . 
\end{align}
For each event, the resulting compound DRep (arising by trying to view Drep 1 and 2 as one entity) has a  probability of yielding a positive weight towards selecting the correct outcome, given by
\begin{align} 
	p_{x_1+x_2}^C =\frac{p_{x_1}p_{x_2}}{P(S)}\ ,  && p_{x_1-x_2}^C =\frac{p_{x_1}(1-p_{x_2})}{P(\overline{S})} \ . 
\end{align}
We next relate the probabilities of the compound DRep to the usual DRep probabilities. 
\begin{lemma}\label{le:compound-probs} 
	For any $x_1\geq x_2$ we have 
	\begin{align}
		p_{x_1+x_2}^C = p_{x_1+x_2} - q_{x_1+x_2} && q_{x_1+x_2} = (x_1+x_2) \cdot \frac{4x_1x_2}{1 + 4x_1x_2}\geq 0  \\
		p_{x_1-x_2}^C = p_{x_1-x_2} - q_{x_1-x_2}  && q_{x_1-x_2} =  (x_1-x_2) \cdot \frac{-4x_1x_2}{1 - 4x_1x_2} \leq 0 
	\end{align}
	thus implying $p_{x_1+x_2}^C \leq  p_{x_1+x_2}$ and $p_{x_1-x_2}^C \geq  p_{x_1-x_2}$.
\end{lemma}
\begin{proof}
	Observe that
	\begin{align}
		p_{x_1+x_2}^C = \frac{p_{x_1}p_{x_2}}{p_{x_1}p_{x_2} + (1-p_{x_1})(1-p_{x_2})} = \frac{1/4+\frac{x_1+x_2}{2}+x_1x_2}{1/2+2x_1x_2} = \frac{1}{2} + \frac{x_1+x_2}{1 + 4x_1x_2}   \\
		p_{x_1-x_2}^C = \frac{p_{x_1}(1-p_{x_2})}{p_{x_1}(1- p_{x_2}) + (1-p_{x_1})p_{x_2}} = \frac{1/4+\frac{x_1-x_2}{2}-x_1x_2}{1/2-2x_1x_2} = \frac{1}{2} + \frac{x_1-x_2}{1 - 4x_1x_2} \ .   
	\end{align}
\end{proof}

We are now in a position to prove the main result of this section. 
\begin{proof}[Proof of Theorem~\ref{thm:psucc-ub-general}]
	We prove by  induction on $n$ that $	\psuccub(x_1,x_2,\ldots,x_n) \leq \psuccub(x_1+x_2+\cdots+x_n)$, which implies the theorem by inequality \eqref{eq:p_succ-ub} and since for a single DRep, $\psuccub(y) = p_y = \psucc(y)$ for any value $y\geq0$.   The base case $n=1$ holds by definition since $\psuccub(x_1)=p_{x_1}$.
	Without loss of generality, assume $x_1\geq x_2$ and let $S$ denote the event that DReps $1$ and $2$ vote in the same way, and by $\overline{S}$ it complement. We denote by $\psuccub(\cdot|E)$ the probability of success conditioned to event $E$. Then, for  $q_S := p^C_{x_1+x_2}$ and $q_{\overline{S}} :=  p^C_{x_1-x_2}$, we have
	\begin{align}\psuccub(x_1,\ldots,x_n) & = \psuccub(x_1,\ldots,x_n|S)P(S) + \psuccub(x_1,\ldots,x_n|\overline{S})P(\overline{S})  \\
		& = \psuccub(C(x_1+x_2,q_S),x_3,\ldots,x_n)P(S) + \psuccub(C(x_1-x_2,q_{\overline{S}}),x_3,\ldots,x_n)P(\overline{S}) \ . 
	\end{align}
	Let $OTHS\geq\alpha$ be the event that the votes of the other DReps, i.e., $3,\ldots,n$  contribute a sum of at least $\alpha$, where we use the interpretation described after \eqref{eq:score-rewritten} that a voter contributes either a positive or negative weight when voting for the correct or the wrong outcome respectively.
	From Lemma~\ref{le:compound-probs}
	\begin{align*}
		\psuccub(C(x_1+x_2),x_3,\ldots,x_n)  & = \\ 
		p_{x_1+x_2}^C\cdot P(OTHS\geq-x_1-x_2) + (1 - p_{x_1+x_2}^C)\cdot P(OTHS\geq x_1+x_2) & =   \\  
		(p_{x_1+x_2} - q_{x_1+x_2})\cdot P(OTHS\geq-x_1-x_2) + (1 - p_{x_1+x_2} + q_{x_1+x_2})\cdot P(OTHS\geq x_1+x_2) & =  \\   
		p_{x_1+x_2} \cdot P(OTHS\geq-x_1-x_2) + (1 - p_{x_1+x_2})\cdot P(OTHS\geq x_1+x_2) & + \\   
		q_{x_1+x_2}\cdot [P(OTHS\geq x_1+x_2) - P(OTHS\geq-x_1-x_2)] & = \\ \psuccub(x_1+x_2,x_3,\ldots,x_n) + q_{x_1+x_2}\cdot [P(OTHS\geq x_1+x_2) - P(OTHS\geq-x_1-x_2)] & \ . 
	\end{align*}
	Similarly,  
	\begin{align*}
		\psuccub(C(x_1-x_2),x_3,\ldots,x_n)  & = \\ 
		p_{x_1-x_2}^C\cdot P(OTHS\geq-x_1+x_2) + (1 - p_{x_1-x_2}^C)\cdot P(OTHS\geq x_1-x_2) & =   \\  
		(p_{x_1-x_2} - q_{x_1-x_2})\cdot P(OTHS\geq-x_1+x_2) + (1 - p_{x_1-x_2} + q_{x_1-x_2})\cdot P(OTHS\geq x_1-x_2) & =  \\   
		p_{x_1-x_2} \cdot P(OTHS\geq-x_1+x_2) + (1 - p_{x_1-x_2})\cdot P(OTHS\geq x_1-x_2) & + \\   
		q_{x_1-x_2}\cdot [P(OTHS\geq x_1-x_2) - P(OTHS\geq-x_1+x_2)] &  = \\   
		\psuccub(x_1-x_2,x_3,\ldots,x_n) + q_{x_1-x_2}\cdot [P(OTHS\geq x_1-x_2) - P(OTHS\geq-x_1+x_2)] & \ . 
	\end{align*}
	Putting things together
	\begin{align}
		\psuccub(x_1,x_2,\ldots,x_n) & = \psuccub(x_1+x_2,x_3,\ldots,x_n)P(S) + \psuccub(x_1-x_2,\ldots,x_n)P(\overline{S}) 
		\nonumber\\
		&  + q_{x_1+x_2}\cdot [P(OTHS\geq x_1+x_2) - P(OTHS\geq-x_1-x_2)] P(S) \nonumber\\ 
		& + q_{x_1-x_2}\cdot [P(OTHS\geq x_1-x_2) - P(OTHS\geq-x_1+x_2)]P(\overline{S}) \nonumber\\
		& \leq \label{eq:inductive-step} \psuccub(x_1+x_2+x_3+\cdots+x_n) -2x_2 P(\overline{S}) 
		\\
		&  + q_{x_1+x_2}\cdot [P(OTHS\geq x_1+x_2) - P(OTHS\geq-x_1-x_2)] P(S) \nonumber\\ 
		& + q_{x_1-x_2}\cdot [P(OTHS\geq x_1-x_2) - P(OTHS\geq-x_1+x_2)]P(\overline{S}) \nonumber  
	\end{align}
	where the inequality follows by the inductive hypothesis. To  complete the proof we need to show that the above quantity is at most $\psuccub(x_1+x_2+x_3+\cdots+x_n)$. It is enough to prove 
	\begin{align}
		q_{x_1+x_2}\cdot [P(OTHS\geq x_1+x_2) - P(OTHS\geq-x_1-x_2)] P(S) + \nonumber \\
		+ q_{x_1-x_2}\cdot [P(OTHS\geq x_1-x_2) - P(OTHS\geq-x_1+x_2)]P(\overline{S})  & \leq 0 \ . \label{eq:two-additional-terms}
	\end{align} 
	Note that we have the following properties
	\begin{align}
		P(OTHS\geq x_1+x_2) -  P(OTHS\geq-x_1-x_2) & = -P(OTHS\in [-x_1-x_2,x_1+x_2)) \\ 
		P(OTHS\geq x_1-x_2) - P(OTHS\geq-x_1+x_2) & =  
		-P(OTHS\in [-x_1+x_2,x_1-x_2)) 
	\end{align}
	and 
	\begin{align}   
		P(S) = 1/2 + 2x_1x_2\ ,  &&  q_{x_1+x_2} = (x_1+x_2) \cdot \frac{4x_1x_2}{1 + 4x_1x_2} = (x_1+x_2) \cdot \frac{2x_1x_2}{P(S)} \ ,  \\ 
		P(\overline{S}) = 1/2 - 2x_1x_2\ ,  && q_{x_1-x_2} =  (x_1-x_2) \cdot \frac{-4x_1x_2}{1 - 4x_1x_2} = (x_1-x_2) \cdot \frac{-2x_1x_2}{P(\overline{S})} \ . 
	\end{align} 
	Hence 
	\begin{align*}
		q_{x_1+x_2}\cdot [P(OTHS\geq x_1+x_2) - P(OTHS\geq-x_1-x_2)] P(S) & = \\ 
		2x_1x_2 (x_1+x_2)[-P(OTHS\in [-x_1-x_2,x_1+x_2)] \ , \intertext{and}\\
		q_{x_1-x_2}\cdot [P(OTHS\geq x_1-x_2) - P(OTHS\geq-x_1+x_2)]P(\overline{S}) & = \\
		2x_1x_2 (x_1-x_2) [P(OTHS\in [-x_1+x_2,x_1-x_2)] & \leq   \\
		2x_1x_2 (x_1+x_2) [P(OTHS\in [-x_1+x_2,x_1-x_2)]
		& \leq   \\
		2x_1x_2 (x_1+x_2) [P(OTHS\in [-x_1-x_2,x_1+x_2)] \ . 
	\end{align*} 
	By putting together these equations we obtain \eqref{eq:two-additional-terms}. The latter together with \eqref{eq:inductive-step} implies   $$\psuccub(x_1,x_2,\ldots,x_n) \leq \psuccub(x_1+x_2+x_3+\cdots+x_n).$$ This completes the proof. 
\end{proof}

\subsection{Proof of Theorem~\ref{thm:concave-convex-lb-effort}}
\begin{proof}
	Let  $c_{lin}(x) := \alpha_{c,B} \cdot x$ be as in Definition~\ref{def:concave-convex-line-general}, that is, 
	\begin{align}
		c(x_1) = B = \alpha_{c,B} \cdot x_1\ , 
	\end{align}
	where $x_1$ as above is the effort that one DRep can use for cost function $c(\cdot)$ and budget $B$. 
	Then, by the concave part of $c(x)$ and by the definition of $x_{int}$, we have
	\begin{align}
		c(x) > c_{lin}(x) = \alpha x && \text{for all } x\in (0,x_{int}) \label{eq:c-general-vs-linear}
		\intertext{and therefore}
		\DRepsopt{c}{x} \leq \DRepsopt{c_{lin}}{x}   && \text{for all } x\in (0,x_{int})  \ . \label{eq:d-general-vs-linear}
	\end{align}
	Hence,  for any $x \in (0,x_{int})$ we have 
	\begin{align*}
		\psuccun(x,\DRepsopt{c}{x}) \leq & \psuccun(x,\DRepsopt{c_{lin}}{x})  && \text{(Lemma~\ref{le:prob-succ-mon} and \eqref{eq:d-general-vs-linear})} \\
		\leq & \psuccun(x_{int},\DRepsopt{c_{lin}}{x_{int}}) && \text{(Lemma~\ref{le:prob-succ-mon} and } x < x_{int}) \\  = & \psuccun(x_{int},1)&& \text{($\DRepsopt{c}{x_{int}} =  \DRepsopt{c_{lin}}{x_{int}}=1$)} 
		\ .  
	\end{align*}
	This completes the proof.
\end{proof}

\subsection{Proof of Theorem~\ref{thm:cocave-convex-tangent}}
\begin{proof}From Definition~\ref{def:points}, for any $x\leq x_{tangent}^c$ we have $c_{lin}(x) \leq c(x)$ and thus $\DRepsopt{c}{x}\leq \DRepsopt{c_{lin}}{x}$, with equality holding for $x=x_{tangent}^c$, that is $\DRepsopt{c}{x_{tangent}^c} =  \DRepsopt{c_{lin}}{x_{tangent}^c}$.  Hence,  
	\begin{align*}
		\psucc(x,\DRepsopt{c}{x}) \leq && \text{(Lemma~\ref{le:prob-succ-mon} and $\DRepsopt{c}{x}\leq \DRepsopt{c_{lin}}{x}$)} \\
		\psucc(x,\DRepsopt{c_{lin}}{x}) \leq && \text{(Assumption~\ref{as:prob-succ-lin-mon} and $x\leq x_{tangent}^c$)} \\ \psucc(x_{tangent}^c,\DRepsopt{c_{lin}}{x_{tangent}^c}) = && \text{($\DRepsopt{c}{x_{tangent}^c} =  \DRepsopt{c_{lin}}{x_{tangent}^c}$)} \\ 
		\psucc(x_{tangent}^c,\DRepsopt{c}{x_{tangent}^c}) \ .  \phantom{=} && 
	\end{align*}
	This completes the proof.
\end{proof}

\subsection{Proof of Theorem~\ref{th_three-vs-one}}
\begin{proof}Since the optimal efforts with one and three DReps, $x^*(1)$ and $x^*(3)$ respectively, satisfy $c(x^*(1))=B=3c(x^*(3))$, we have $x^*(3) = \frac{x^*(1)}{3^{1/\beta}}$. Condition \eqref{eq:three-vs-one} in Lemma~\ref{le:three-vs-one} boils down to 
    \begin{align}\label{eq:three-vs-one:monomials}
   x^*(1)<\frac{3\frac{x^*(1)}{3^{1/\beta}} - 4 \frac{x^*(1)^3}{3^{3/\beta}}}{2}  &&
   \Leftrightarrow 
   && 
   2 - 3\frac{1}{3^{1/\beta}} < - 4 \frac{x^*(1)^2}{3^{3/\beta}}
   &&
   \Leftrightarrow 
   && 
     x^*(1)^2 < \frac{3^{3/\beta}}{4} (3^{1 - 1/\beta}  - 2) 
     \ . 
\end{align}

 For $\beta <\beta^\star$, we have $3^{1 - 1/\beta} < 3^{1 - 1/\beta^\star} =  3^{\ln(2)/\ln(3)}=2$, thus implying that the inequality above cannot hold. In particular, for all $x^*(1)$ the opposite inequality holds strictly and thus (Lemma~\ref{le:three-vs-one}) one DRep is better that three DReps.    This proves the first part of the theorem. 

 For $\beta > \beta^\star$, the inequality above holds and it actually gives
 \begin{align}
    x^*(1) < \frac{1}{2}\sqrt{3^{3/\beta} (3^{1 - 1/\beta}  - 2)} 
 \end{align}
 which combined with the identity $B=c(x^*(1))=x^*(1)^\beta $ yields
 \begin{align}
    B^{1/\beta} < \frac{1}{2}\sqrt{3^{3/\beta} (3^{1 - 1/\beta}  - 2)} && \Leftrightarrow && B < 3^{3/2} \cdot \left(\frac{3^{1 - 1/\beta}  - 2}{4}\right)^{\beta/2} = B^\star
 \end{align}
which proves the second part of the theorem. This completes the proof.
\end{proof}

\subsection{Proof of Lemma~\ref{le:inflection-S-shape}}
\begin{proof}
	As shown in \cite{murre2014s}, the learning curve in \eqref{eq:learning-S-shape-exp} has an inflection point at $\Tilde{t} = \ln(\xi)/\mu$, for all $\xi>1$. Thus, the corresponding inflection point $\Tilde{x}$ for the cost function \eqref{eq:cost-S-shaped-exp} corresponds to 
	\begin{align}
		2\Tilde{x} = p_\xi(\Tilde{t}) = [1 - \exp(-\mu \Tilde{t})]^\xi = [1 - \exp(-\ln(\xi))]^\xi = [ 1 - 1/\xi]^\xi 
	\end{align}
	which proves the statement regarding the inflection point. To see why $c_{\xi}$ is concave-convex we observe the following. The function $p_{\xi}(t)$ is shown to be strictly increasing \emph{convex-concave} in its argument $t$ in \cite{murre2014s}. Then the inverse function $p_{\xi}^{-1}(p)$ exists and is \emph{concave-convex} in $p\in(0,1)$. Finally, our cost function $c_\xi(x)$ satisfies $p=1/2 + x$ and $t=c_\xi(x) = p_{\xi}^{-1}(1/2+x)$, and thus $c_\xi(x)$ is concave-convex in $x\in(0,1/2)$. 
\end{proof}

\subsection{Proof of Lemma~\ref{le:opt-efforts-S-shape}}
\begin{proof}
	The optimal $x^*(k)$ is given by the  $x$ such that $k \cdot c(x)=B$, which for this particular cost function boils down to satisfy
	\begin{align}
		c(x) = -\frac{1}{\mu}\ln (y) = B/k\ ,  && y = 1 - (2x)^{1/\xi} \ . 
	\end{align}
	By rearranging the terms
	\begin{align}
		\exp(-B\mu/k) = y = 1 - (2x)^{1/\xi} && \Rightarrow &&
		2x = [1 - \exp(-B\mu/k)]^\xi   
	\end{align}
	and the latter equation together with \eqref{eq:learning-S-shape-exp} yields the desired result on $x^*(k)$.
\end{proof}

\subsection{Proof of Corollary~\ref{cor:S-shape-one-vs-three}}
We next quantify the optimal symmetric efforts for this family of functions.

\begin{lemma}\label{le:opt-efforts-S-shape}
    For the cost function $c_{\xi}$  in \eqref{eq:cost-S-shaped-exp} the optimal symmetric effort $x^*(k)$, given a budget $B$, is equal to 
    $
        x^*(k) = \frac{p_\xi(B/k)}{2} = \frac{[1-\exp(-B\mu/k)]^\xi}{2} \ . 
    $
\end{lemma}
We are now in a position to  apply the results in Section~\ref{sec:concave-convex} to this specific class of concave-convex cost functions parameterized in $\xi>1$, the complexity parameter. 
\begin{proof}[Proof of Corollary~\ref{cor:S-shape-one-vs-three}]
	The first part follows from Corollary~\ref{cor:geometric}, with the observation that $B^* = c(x_{tangent}^c)> c(x_{inflection}^c)=\ln(\xi)/\mu$, since $x_{tangent}^c>x_{inflection}^c$ and because of Lemma~\ref{le:inflection-S-shape}. The second part follows from Lemma~\ref{le:opt-efforts-S-shape}.
\end{proof}

\subsection{Proof of Theorem~\ref{th:S-shape-euilibria}}
\begin{proof} We rewrite the condition on the effort $x$ in Theorem~\ref{thm:top-k}, and apply the definition of our specific cost function. 
For the case $\frac{k}{k-1}x\in (0, 1/2]$, the conditions are: 
    \begin{align*}
    c\left(\frac{k}{k-1}x\right) \geq B/k \geq c(x) & \Longleftrightarrow \\
    -\frac{1}{\mu}\ln \left(1 - (\frac{2k}{k-1}x)^{1/\xi}\right) \geq B/k \geq -\frac{1}{\mu}\ln \left(1 - (2x)^{1/\xi}\right) & \Longleftrightarrow \\
    \ln \left(1 - (\frac{2k}{k-1}x)^{1/\xi}\right) \leq -\mu B/k \leq \ln \left(1 - (2x)^{1/\xi}\right) & \Longleftrightarrow \\
    1 - (\frac{2k}{k-1}x)^{1/\xi} \leq \exp(-\mu B/k) \leq 1 - (2x)^{1/\xi} &  
\end{align*}
that is 
\begin{align}\frac{k-1}{k}\cdot \frac{[1 - \exp(-\mu B/k)]^k}{2} \leq  x
&& \text{and} && 
       x \leq \frac{[1 - \exp(-\mu B/k)]^\xi}{2} \ . 
\end{align}
This and Lemma~\ref{le:opt-efforts-S-shape} yield the desired result in this case.

Finally, for the  case $\frac{k}{k-1}x> 1/2$ we have only the right inequality above. 
\end{proof}

\section{Evidence for Assumption~\ref{as:prob-succ-lin-mon}}\label{app:assumption-weaker}
In this section, we provide some evidence in support of Assumption~\ref{as:prob-succ-lin-mon}.  
Obviously,  Theorem~\ref{thm:psucc-ub-general} says that Assumption~\ref{as:prob-succ-lin-mon} holds in the special case of $k'=1$. In further support of Assumption~\ref{as:prob-succ-lin-mon}, we prove the following weaker result saying that monotonicity holds in the following case. First, we consider the case in which, instead of reducing the number of DReps, we \emph{halve} them. That is, we have $k'=2k$ in Theorem~\ref{thm:psucc-ub-general}. Furthermore,  we consider  the variant in which, when the votes lead to a tie, the correct outcome is selected with probability $1$, instead of $1/2$. Note that is the same variant used in the proof of Theorem~\ref{thm:psucc-ub-general}, and the corresponding probability is denoted as $\psuccub(\cdot)$. The next theorem states that under these two conditions, monotonicity hold. 

\begin{theorem}[weaker version of Assumption~\ref{as:prob-succ-lin-mon}]\label{thm:as-prob-succ-lin-mon}
	For any linear cost function $c_{lin}(x) = a\cdot x$ with $a>0$, the following holds.  For any $x \leq x'$
	\begin{align}
		\psuccub(x,k) \leq \psuccub(x',k') 
	\end{align}
whenever $k = \DRepsopt{c_{lin}}{x}$ and $k' = \DRepsopt{c_{lin}}{x'}$ satisfy $k=2k'$. 
\end{theorem}
\begin{proof}
Let $\vec{x}$ and $\vec{x'}$ denote the optimal symmetric effort vectors in which $k$ and $k'$ DReps exert positive effort, for the variant with success probability $1$ in case of ties. That is,  $\psuccub(x,k) =\psuccub(\vec{x})$ and $\psuccub(x',k') =\psuccub(\vec{x'})$. We consider the random variables in \eqref{eq:score-rewritten} corresponding to these two optimal solutions, 
\begin{align}
    Z_1,Z_2,\ldots,Z_{k} && Z'_1,Z'_2,\ldots,Z'_{k'}
\end{align}
where $\prob{Z_i=x} = 1/2 + x$ and $\prob{Z'_i=x'} = 1/2 + x'$. Our goal is to prove the following inequality:
\begin{align}
     \psuccub(\vec{x}) = \prob{\sum_{i=1}^{k} Z_i \geq  0} \leq  \psuccub(\vec{x'}) = \prob{\sum_{i=1}^{k'} Z'_i \geq  0}
\end{align}
In order to upper bound $\psuccub(\vec{x})$, we group the $k=2k'$  random variables into $k'$ pairs,  forming $k'$ \emph{new} random variables, as follows: 
\begin{align}
\label{eq:def-yi-pair}
    Y_i := \frac{Z_{2i - 1} +  Z_{2i}}{2} \ . 
\end{align}
These new random variables take values in $\{x,0,-x\}$ and 
\begin{align}\nonumber
  \prob{Y_i =  x | Y_i \neq 0} = & \frac{(1/2 + x)^2}{1 - 2(1/2 + x)(1/2-x)} 
  = \frac{4x^2+1 + 4x}{2(4x^2+1))} 
  \\ \label{eq:pair-variables}
  = & \frac{1}{2} + \frac{2x}{(2x)^2 + 1} = \frac{1}{2} + \frac{x'}{(x')^2 + 1} \leq \frac{1}{2} + x' = \prob{Z_i'=x'} \ . 
\end{align}
Let $E_\ell$ denote the event that $\ell$ out of the $Y_i$'s are equal to 0, and observe that \eqref{eq:pair-variables}  implies 
\begin{align}\label{eq:pairs-prob-vs-variables}
    \prob{\sum_{i=1}^k Y_i \geq 0 \mid  E_\ell} \leq \prob{\sum_{i=1}^{k'} Z_i' \geq 0}=\psuccub(\vec{x'})\ . 
\end{align}
Therefore, 
\begin{align}
\label{eq:pair-s-prob-ub}
    \prob{\sum_{i=1}^k Y_i \geq 0} = \sum_{\ell=1} ^k \prob{\sum_{i=1}^k Y_i \geq 0 | E_\ell} \cdot \prob{E_\ell}
    \stackrel{\eqref{eq:pairs-prob-vs-variables}}{\leq}
    \psuccub(\vec{x'})\cdot 
    \sum_{\ell=1} ^k \prob{E_\ell} =
    \psuccub(\vec{x'}) \ . 
\end{align}
To complete the proof we observe that
\begin{align}
    \psuccub(\vec{x}) = \prob{\sum_{i=1}^{2k} Z_i \geq 0} \stackrel{\eqref{eq:def-yi-pair}}{=} \prob{\sum_{i=1}^k Y_i \geq 0} \stackrel{\eqref{eq:pair-s-prob-ub}}{\leq} \psuccub(\vec{x'}) 
\end{align}
which completes the proof. 
\end{proof}

\end{document}